\documentclass[11pt,twoside]{article}  

\usepackage[utf8]{inputenc}
\usepackage[margin = 2.5cm]{geometry}
\usepackage{amsmath,amssymb,amsfonts,amsthm}
\usepackage{appendix}
\usepackage[dvipsnames]{xcolor}
\usepackage{siunitx}
\usepackage{graphicx}
\usepackage[numbers]{natbib}
\RequirePackage[colorlinks,citecolor=blue,urlcolor=blue, allcolors=blue]{hyperref}
\usepackage{newpxtext,newpxmath} %

\newcommand*\diff{\mathop{}\!\mathrm{d}}

\usepackage{enumitem}
\usepackage{subcaption}
\usepackage{xspace}
\usepackage{listings}

\usepackage{booktabs}

\theoremstyle{plain}
\newtheorem{prop}{Proposition}
\newtheorem{lem}{Lemma}
\newtheorem{thm}{Theorem}

\theoremstyle{definition}
\newtheorem{defn}{Definition}

\newtheorem{examplex}{Example}
\newenvironment{ex}
  {\pushQED{\qed}\examplex}
  {\popQED\endexamplex}

\newcommand{\R}{\mathbb{R}}
\newcommand{\N}{\mathbb{N}}
\renewcommand{\P}{\mathbb{P}}
\newcommand{\E}{\mathbb{E}}

\newcommand{\mc}{\mathcal}
\newcommand{\msf}{\mathsf}
\newcommand{\mb}{\mathbb}
\newcommand*{\from}{\colon}

\allowdisplaybreaks

\title{Using Participants' Utility Functions\\ to Compare Versions of Differential Privacy}

\author{Nitin Kohli and Michael Carl Tschantz\thanks{This work was primarily conducted while the first author was at the School of Information, University of California, Berkeley, and the second author was at the International Computer Science Institute.  The second author gratefully acknowledges funding support from
the National Science Foundation
(Grant~1704985). %
The opinions in this paper are those of the authors and do not necessarily reflect the opinions of any funding sponsor or the United States Government.}}

\date{}

\begin{document}

\maketitle

\begin{abstract}
We use decision theory to compare variants of differential privacy from the perspective of prospective study participants. We posit the existence of a preference ordering on the set of potential consequences that study participants can incur, which enables the analysis of individual utility functions. Drawing upon the theory of measurement, we argue that changes in expected utilities should be measured via the classic Euclidean metric. We then consider the question of which privacy guarantees would be more appealing for individuals under different decision settings. Through our analysis, we found that the nature of the potential participant’s utility function, along with the specific values of $\epsilon$ and $\delta$, can greatly alter which privacy guarantees are preferable. 
\end{abstract}

\section{Introduction}
\label{sec:intro}

In 2006, Dwork, McSherry, Nissim, and Smith introduced the concept of $\epsilon$-indistinguishability to enable the computation of privacy-preserving statistics \cite{dwork06crypto}. This concept -- which would be renamed $\epsilon$-differential privacy  within the same year \cite{dwork06icalp} -- is a  precise mathematical property of an algorithm that bounds how their outputs may vary as their inputs vary. The term ``differential privacy'' sometimes refers to this precise property, but also refers a series of related properties that extend or alter the original one~\cite[e.g.,][]{dwork06eurocrypt, mironov2017renyi, dwork2016concentrated, bun2016concentrated}. 

These variants of differential privacy vary in how they bound the amount that the outputs may change in response to changes in the inputs. The variants most related to the original $\epsilon$-differential privacy are each a way of bounding the effect size of each input~\cite{tschantz20sp}, and it is on these that we will focus. Given the long history of debates over the correct measure of effect sizes~\cite[e.g.,][]{sormani17multi}, it should come as no surprise that there's no consensus on which versions of differential privacy are the best, or even agreement on what would make one best.

There are different points of view that one can adopt when comparing different definitions of differential privacy. For example, Beimel~et~al.\@ compare $\epsilon$-differential privacy and $(\epsilon,\delta)$-differential privacy from the perspective of sample complexity~\cite{beimel14tc}. Another perspective is to frame comparisons in terms of a \textit{privacy-``utility''} tradeoff \cite{dwork2019differential}, where the notion of utility is in reference to the ``utility of an analysis’’ from the surveyor’s perspective.\footnote{Given the polysemous nature of the word ``utility,'' others have referred to this as the privacy-accuracy tradeoff to more precisely convey that this tradeoff effects the accuracy of the statistic \cite{blumenstock2023big}. Others have refined the concept of the privacy-accuracy tradeoff to settings where the accuracy of the interventions based on a statistic are more important that the accuracy the statistic itself, called the privacy-intervention accuracy tradeoff \cite{kohli2023privacy}.} 

In our inquiry, we consider how different definitions of differential privacy compare from the point of view of potential study participants. We approach this question using the tools of decision theory. In particular, we look at how the properties of a survey participant's utility functions affect which of absolute $\epsilon$-differential privacy,  $(\epsilon, \delta)$-approximate differential privacy, and absolute $\delta$-differential privacy would be more appealing to them, where we use the last term to refer to $(0, \delta)$-approximate differential privacy. 

The structure of our analysis is as follows. We posit the existence of a preference ordering on the set of potential consequences that study participants can incur, which enables the analysis of individual utility functions. Drawing upon the theory of measurement, we argue that changes in expected utilities should be measured via the classic Euclidean distance. With this measurement tool in hand, we set out to explore how different variants of differential privacy effect participants' utility functions. Through our analysis, we found that the nature of the potential participant’s utility function can alter which privacy guarantees would be more appealing.

When individual utility functions are real-valued and bounded, we show there are situations where individuals' may forgo joining a pure $\epsilon$-differentially private study. In such situations, there is a compensation scheme that induces participation. Interestingly, under this compensation scheme we find that there are values of $\epsilon, \epsilon', \delta' > 0$ such that inducing participation with pure $\epsilon$-differential privacy is more expensive that using $(\epsilon',\delta')$-differential privacy. 

However, when individuals' utility functions are allowed to take hyperreal values, we are able to construct decision problems in which no amount of compensation can encourage participation. On the other hand, we are also able to construct decision problems using hyperreal utility functions where individuals will only take part in a study when pure $\epsilon$-differential privacy is used. Taken together, the examples and analysis in this manuscript speak to the context-specific dependencies of using utility theory to understand the effects that privacy-enhancing technologies have on participants. For this reason, our work can be seen as contributing a lineage of research documenting the nuanced relationship between privacy and economics ~\cite{acquisti2016economics, laskowski2014government, johnson2015caviar, acquisti2022economics}. 

\section{Mathematical Preliminaries}
\label{sec:preliminaries}

In this section, we introduce the mathematical primitives and notation that we will utilize throughout this study. We begin by introducing differential privacy, and recall some fundamental results of particular interest for this study. We then present a deep-dive in the fundamentals of preference and utility theory to set the stage to clearly and rigorously study differential privacy from the perspective of utility theory.

\subsection{Background on Differential Privacy}
\label{sec:prelim_dp}

The only use case for differential privacy that we consider involves reporting the results of a survey.  It envisions a surveyor attempting to collect survey responses from people who may either provide them or not.  One reason a potential participant might decline to participate is privacy concerns about how their response might be reveal private information, even when aggregated with other participants' responses to form statistics \cite{dwork2017exposed, kohli2021leveraging}. 
To address this concern, the surveyor may promise to release only a differentially private statistic of dataset.  Conceptually speaking, the value of the statistic will not be affected much by whether any single person provides data. For this reason, it has been previously argued in that  each potential participant need not worry about how participating will affect their privacy and might as well participate \cite{dwork2014algorithmic}. Although in Section~\ref{sec:inf}, we'll consider some complications.

Consider a collection of $n$ agents, denoted as $[n] = \{1,...,n\}$. As is common in the game theory literature between computer science and economics, we will refer to agents as ``players'' interchangeably. 

For each agent $j$, let $\mc{X}_j$ denote the set of values they can report to a mechanism $\mc{M}$. We will assume $\mc{X}_j$ is finite and contains a null value $\perp$ that represents opting-out. Let $\mc{X} = \Pi_{j \in [n]} \mc{X}_j$ denote the space of inputs to the mechanism $\mc{M}$, and let $\mc{O}$ be the output space. Each $x \in \mc{X}$ can be represented by an $n$ dimensional vector $x = (x_1,...,x_n)$. We will say that two vectors $x, x' \in \mc{X}$ are \emph{neighboring} vectors if there exists a unique agent $j \in [n]$ such that $x_j \ne x'_j$ and $x_i = x'_i$ for all $i \ne j$.\footnote{In differential privacy, there is an alternative notion of neighboring datasets (referred to as \emph{unbounded}~\cite{kifer2011no} or \emph{presence--absence neighbors}~\cite{kohli2021leveraging}), which requires that one of $x$ and $x'$ must be $\bot$. The notion we choose in this analysis is often referred to as \emph{bounded}~\cite{kifer2011no} or \emph{switch neighbors}~\cite{kohli2021leveraging}. This difference doesn't affect the points we make herein.} That is, $x$ and $x'$ are neighbors if there is exactly one player whose value differs (namely between $x_j$ and $x'_j$). Let 
$$
\mc{X}_{-j} = \Pi_{i \in [n]-\{j\}} \mc{X}_i
$$
represent the space of all the messages that any player, other than $j$, can report to the mechanism. For any $x \in \mc{X}$ we denote the inputs by all players other than player $j$ as 
$$x_{-j} = (x_1,...,x_{j-1}, x_{j+1}, ..., x_{n}) \in \mc{X}_{-j}$$ 
We will refer to this as an \emph{environment} that agent $j$ operates within~\cite{kohli2018epsilon}, and write $x = (x_j, x_{-j})$.

A mechanism $\mc{M}: \mc{X} \rightarrow \mc{O}$ is a randomized mapping from $\mc{X}$ to $\mc{O}$. A particular subclass of interest to this study are \emph{differentially private mechanisms}. Intuitively speaking, a mechanism $\mc{M}$ satisfies differential privacy if the probability of any output from $\mc{M}$ is ``not too sensitive'' on any single component of $x \in \mc{X}$~\cite{kohli2023differential}.

\begin{defn} ($\epsilon$-Differential Privacy, \cite{dwork06crypto, dwork06icalp})
A randomized mechanism $\mc{M} \from \mc{X} \to \mc{O}$ satisfies \emph{$\epsilon$-differential privacy} if, for all neighboring $x,x' \in \mc{X}$ and for all measurable sets $\mc{S} \subseteq \mc{O}$,
$$
\P(\mc{M}(x) \in \mc{S}) \le e^{\epsilon} \P(\mc{M}(x') \in \mc{S})
$$
\end{defn}

An alternative way to view this is in terms of the reports of other agents. Suppose that $x$ and $x'$ are neighboring inputs that differ in the $j^{th}$ component. Then we can write $x = (x_j, x_{-j})$ and $x = (x'_j, x_{-j})$ for some $x_j, x'_j \in \mc{X}$ and $x_{-j} \in \mc{X}_{-j}$.
Mechanisms that are $\epsilon$-differentially private guarantee that, regardless of the environment $x_{-j}$, the choice of reporting $x_j$ or $x'_j$ doesn't alter the chance of a mechanism output by ``too much,'' where ``too much'' is made precise by the term $\epsilon$. This holds true whether the report $x_j$ is the opt-out message $\perp$ or any other value in $\mc{X}_j$. Differential privacy ensures that the output of a mechanism isn't too sensitive based on any single message $x_j$ sent to the mechanism. The parameter $\epsilon$ can be viewed as the size effect of an individual's message on the output in the worst-case environment~\cite{tschantz20sp}.

A useful relaxation of $\epsilon$-differential privacy is $(\epsilon,\delta)$-differential privacy~\cite{dwork06eurocrypt}. Mechanisms $\mc{M}$ that satisfy \emph{$(\epsilon,\delta)$-differential privacy} ensures that the probabilistic inequality from $\epsilon$-differential privacy holds for $S \subseteq \mc{O}$, but with some slippage $\delta \in [0,1]$.

\begin{defn} ($(\epsilon,\delta)$-Differential Privacy, \cite{dwork06eurocrypt})
A randomized mechanism $\mc{M} \from \mc{X} \to \mc{O}$ satisfies $(\epsilon, \delta)$-differential privacy if, for all neighboring $x,x' \in \mc{X}$ and for all measurable sets $\mc{S} \subseteq \mc{O}$,
$$
\P(\mc{M}(x) \in \mc{S}) \le e^{\epsilon} \P(\mc{M}(x') \in \mc{S}) + \delta
$$
\end{defn}

At times, we will refer to $(\epsilon, 0)$-differential privacy as \emph{pure $\epsilon$-differential privacy}. At the other end of the spectrum, we will refer to $(0,\delta)$-differential privacy as \emph{pure $\delta$-differential privacy}. And lastly, in the case where $\epsilon > 0$ and $\delta > 0$, we will refer to $(\epsilon, \delta)$-differential privacy as \emph{approximate differential privacy}. 

One of the main strengths of differential privacy is its robust guarantees \cite{kroll2019privacy}. Of particular interest to this study is the post-processing guarantee.

\begin{lem} (Post-Processing, Proposition 2.1 in \cite{dwork2014algorithmic})
Let $f:\mc{O} \rightarrow \mc{C}$ be a (possibly randomized) function. If $\mc{M}: \mc{X} \rightarrow \mc{O}$ satisfies $(\epsilon,\delta)$-differential privacy, then so too does $f \circ \mc{M}$.
\end{lem}

For notation simplicity, we will refer to $f\circ \mc{M}$ as $f_{\mc{M}}: \mc{X} \rightarrow \mc{C}$. So as long as $\mc{M}$ is differentially private with output space $\mc{O}$, the post-processing result guarantees that $f_{\mc{M}}$ is differentially private over this new space $\mc{C}$.

\subsection{Background on Utility Theory}
\label{sec:prelim_utility}

Next, we describe the mathematical preliminaries of utility theory. We consider a simple model, where each agent $j \in [n]$ has a \emph{preference} over a countable set of consequences $\mc{C}$. Mathematically, this preference is formulated by positing the existence of a binary relation $\succeq_j \subseteq \mc{C} \times \mc{C}$ for all $j \in [n]$ that is complete and transitive.

\begin{defn}
Consider a binary relation $\succeq_j \subseteq \mc{C} \times \mc{C}$. We say this relation is \emph{complete} if for all $a,b \in \mc{C}$, we have $a \succeq_j b$ or $b \succeq_j a$. We say a binary relation is \emph{transitive} if for all $a,b,c \in \mc{C}$, if $a \succeq_j b$ and $b \succeq_j c$, then $a \succeq_j c$. When $\succeq_j$ is complete and transitive, we say it is a \emph{preference ordering} (or simply a \textit{preference} as shorthand).
\end{defn}

For $c_1,c_2 \in \mc{C}$, we encode the idea that agent $j$ ``prefers consequence $c_1$ at least as much as consequence $c_2$'' by $(c_1,c_2) \in \succeq_j$, which we will often write as $c_1 \succeq_j c_2$. 

    Two remarks are in order. First, there will be times in our analysis where we require $\mc{C}$ to be finite, and not countably infinite. When this is the case, we will state it clearly inline. 
    
    Second, our choice to model preferences \emph{only} over consequences may seem narrow. Indeed, it is possible that individuals have preferences over a wider collection of things, such as preferences over different types of privacy protection~\cite{kohli2018epsilon, cummings2016empirical}. There are two reasons we consider this simple type of preference. The first reason as a point of comparison with prior work in differentially private mechanism design, where agents' utility function are typically modeled only as function the outcomes or consequences of a computation~\cite{dwork2014algorithmic, mcsherry07focs, hsu2014differential, pai2013privacy}. The second reason is that, even when preferences are construed narrowly, the results we derive are nonetheless nuanced. As such, they provide useful evidence about the contextual nature of results that connect privacy to utility theory.

Since preference relations are defined in terms of order theory, they can be quite difficult to work with numerically. As such, we will work with real-valued \textit{utility representations} $u_j: \mc{C} \rightarrow \R$ of these preference relations $\succeq_j$. 
(Later we'll consider functions going beyond $\R$, such as the hyperreals $^{\star}\R$.)
To ensure that a utility representation faithfully encodes the ordinal information from $\succeq_j$, we say a utility function $u_j$ \textit{represents} $\succeq_j$ if, for all $c_1,c_2 \in \mc{C}$,
$c_1 \succeq_j c_2 \iff u_j(c_1) \ge u_j(c_2)$. For countable $\mc{C}$, the assumptions that $\succeq_j$ is complete and transitive guarantees the existence of $u_j$. 

Utility functions are analytic devices that facilitate numerical reasoning about preferences. While this is a useful tool, care must be taken in practice, as multiple utility functions can represent the same preference $\succeq_j$. That is, if $u_j$ represents $\succeq_j$ and $h$ is any monotonically increasing function, then $h \circ u_j$ also represents $\succeq_j$. As such, utility representations are not unique, so caution must be taken when interpreting the specific numerical values assigned by utility functions.

The preference framework above was developed to reason about an agent's choice under certainty. When an agent faces probabilistic consequences, such as from a randomized mechanism $\Lambda$, we evaluate the agent's expected utility \cite{giocoli1998true, jensen1967introduction}. In principle, the behavior of the $\Lambda$ can change based on any of the agents' reports. To encode this behavior, suppose that each agent takes reports some $x_j \in \mc{X}_j$ to some mechanism $\Lambda$. For $x = (x_1,...,x_n)$, we write $\P(\Lambda(x) = c)$ to represent the probability that $\Lambda(x)$ yields $c$. With this machinery, we can now define expected utility.

\begin{defn}
    For an agent with utility $u_j$ over $\mc{C}$, the \textit{expected utility} of $x_j$, when other agents report $x_{-j}$, under $\Lambda(x_j, x_{-j})$ is given by
$$
\E_{c \sim \Lambda(x_j, x_{-j})}[u_j(c)] = \sum_{c \in \mc{C}} u_j(c)\P(\Lambda(x_j, x_{-j}) = c) 
$$
\end{defn}

Note that expected utility is defined with respect to a utility function $u_j$ and not the individual's preference $\succeq_j$. As such, different utility representations of $\succeq_j$ can lead to different numerical quantities. This is not necessarily an issue, provided that care is taken when interpreting expected utilities. 

In particular, we must restrict transformations of utility functions to positive affine transformations (and not monotonically increasing functions in general) -- otherwise, the expected utility formulation can provide inconsistent results \cite{briggs2023normative}. In the language of measurement theory, this means that our utility representations of $\succeq_j$ are interval measurements \cite{schoemaker1982expected}. That is, $u_j$ represents $\succeq_j$ $\iff$ $v_j(\cdot) = au_j(\cdot) + b$ also represents $\succeq_j$ for every $a > 0$ and for every $b \in \R$. 

\section{Connections Between Utility Theory and Differential Privacy}
\label{sec:connections}

We now define the process that underlies our analysis. Each agent $j \in [n]$ is given the opportunity to take part in a study. Agent $j$ submits a value $x_j \in \mc{X}_j$ to the mechanism $\mc{M}$ (which could be the message of opting-out of the study: $x_j = \perp$), which then processes the data as $\mc{M}(x_1,\ldots,x_n)$. This produces an outcome $o \in \mc{O}$. Based on the outcome, some decision is made using $o$ that has consequences for the agents, modelled by a (possibly randomized) function $f(o) = c \in \mc{C}$. Agents then experience some utility $u_j(c)$ based on the consequence $c$. Therefore, agent $j$'s expected utility under $f_{\mc{M}}(x)$ is
$$
\E_{c \sim f_{\mc{M}}(x)}[u_j(c)] = \sum_{c \in \mc{C}} u_j(c)\P(f_{\mc{M}}(x) = c)
$$
As such, we can reason about an agent's expected utility under $f_{\mc{M}}(x)$ by examining the convex combination of the utilities induced by their preference $\succeq_j$ and the coin-flips of $f_{\mc{M}}(x)$.

\subsection{The Ratio-Scale: A Candidate for Expected Utility Comparisons}
\label{subsec:ratioscale}

Under suitable assumptions on the utility function $u_j$, the pure differential privacy inequality can be equivalently represented in terms of expected utility. The following is a classic result \cite{dwork2014algorithmic}, which we restate below in the notation of this study.

\begin{prop} (Ratio-Scale Utility Bound, Claim 4 of Roth\footnote{See \href{https://www.cis.upenn.edu/~aaroth/courses/slides/privacymd/Lecture1.pdf}{Intro to Differential Privacy for Game Theorists, and Digital Goods
Auctions, Lecture 1} by Professor Aaron Roth for proof of this statement. Because of this equivalence, some authors use this expected value characterization in place of the definition of $\epsilon$-differential privacy ~\cite[e.g.,][]{pai2013privacy}.})
\label{prop:ratio}
A mechanism $\mc{M}:\mc{X} \rightarrow \mc{O}$ satisfies $\epsilon$-differential privacy if and only if for all utility functions $u_j \from \mc{C} \to \R_{\ge 0}$, for all $f: \mc{O} \rightarrow \mc{C}$, for all $x,x' \in \mc{X}$ neighboring vectors, and for all $j \in [n]$, 
$$
\E_{c \sim f_{\mc{M}}(x_j, x_{-j})}[u_j(c)] \le e^{\epsilon}\E_{c \sim f_{\mc{M}}(x'_j, x_{-j})}[u_j(c)] 
$$
\end{prop}

This proposition can be viewed as a ratio-scale for measuring changes in expected utilities. Namely, whenever both $\E_{c \sim f_{\mc{M}}(x_j, x_{-j})}[u_j(c)] \ne 0$ and $\E_{c \sim f_{\mc{M}}(x'_j, x_{-j})}[u_j(c)] \ne 0$, the symmetry of Proposition \ref{prop:ratio} implies

$$
e^{-\epsilon} \le \frac{\E_{c \sim f_{\mc{M}}(x_j, x_{-j})}[u_j(c)]}{\E_{c \sim f_{\mc{M}}(x'_j, x_{-j})}[u_j(c)]} \le e^{\epsilon}
$$

One would hope that we could use a ratio quantity like this to study preferences. However, depending on the range of utility functions in consideration, this may not be advisable. In this manuscript, we want our analysis to consider as wide an array of utility functions as possible. For this reason, we want $\succeq_j$ to be representable by functions that map to any subset of $\R$, and not just $\R_{\ge 0}$. With this in mind, we find in Example \ref{ex:hcua} that this inequality does not necessarily hold for $u_j$ that can attain both positive and negative values. This presents an issue for our analysis, as we need our comparison of expected utilities to draw the same conclusions for every pair of utility functions real-valued $u_j$ and $v_j$ that represent $\succeq_j$. 

\begin{ex}
\label{ex:hcua}
Suppose the House Committee on Un-American Activities (HCUA) conducts of a survey to determine the number of political dissidents in the Screen Writers Guild (SWG).
However, HCUA, attempting to appear to respect privacy, will have a trusted natural party use differential privacy to compute a noisy count of the number of political dissidents.

One member of SWG, John, is the most suspected of being a political dissident.
HCUA, not known for its reasonable tactics, will cause problems for John if even just a single member of SWG is found to be a political dissident under the presumption that it is John.
In principle, the amount grief that they give John (and therefore John's utility for each noisy count value that the algorithm might produce) should depend upon the expectations about the real (pre-noise) number of political dissidents.
However, the subtleties of differential privacy is largely lost on HCUA, who plan to give John a fixed amount of grief if and only if the noisy count is non-zero.

We can model this example as a game.  For simplicity, we presume that SWG has only two members who play in the game.  Thus, the count is bounded to be no more than $2$ even if both members participate. We let the outcome space be $\mc{O} = \{0,1,2\}$. The possible consequences for John are grief $\msf{g}$ or no-grief $\msf{n}$. As such, we take $\mc{C} = \{\msf{g},\msf{n}\}$.

Each player can send a message $x_i \in \{0,1\}$ to a mechanism $\mc{M}$ with $1$ indicating being a political dissident. Let $\mc{M}$ be a differentially private mechanism that works as follows: The mechanism computes $x_1+x_2$, and returns a value $z \in \mc{O}$ with probability $p_z$ according to the probability distribution below.
\begin{center}
\renewcommand{\arraystretch}{1.3}
 \begin{tabular}{@{}c c c c@{}} 
 \toprule
 $x_1+x_2$ & $p_0$ & $p_1$ & $p_2$ \\
 \midrule
  $0$ & $\varphi$ & $e^{\epsilon}\varphi$ &  $e^{2\epsilon}\varphi$ \\ 
  $1$ & $\varphi$ & $e^{2\epsilon}\varphi$ & $e^{\epsilon}\varphi$ \\
  $2$ & $e^{\epsilon}\varphi$ & $e^{2\epsilon}\varphi$ & $\varphi$ \\
 \bottomrule
\end{tabular}
\end{center}
where $\varphi = (1+e^{\epsilon}+e^{2\epsilon})^{-1}$ and $\epsilon > 0$. 
Note that for all possible values of $x_1+x_2$ the mechanism indeed induces a valid probability distribution, as each $p_z > 0$ and
$$
p_0 + p_1 + p_2 =\varphi + e^{\epsilon}\varphi + e^{2\epsilon}\varphi = (1 + e^{\epsilon} + e^{2\epsilon})\varphi = 1
$$
By inspection, one can verify that $\mc{M}$ is $\epsilon$-differentially private. Consider the two states of the world: one where $(x_1,x_2) = (1,0)$ and another where $(x_1,x_2) = (1,1)$.

We let John be player two and use the utility function given by

$$
v_2(c) = 
\begin{cases} 
  2, & c=\msf{n} \\
  0, & c=\msf{g} \\
\end{cases}
$$
Then by Proposition \ref{prop:ratio}, we would conclude $\E_{c \sim f_{\mc{M}}(1,1)}[v_2(c)] \leq e^{\epsilon}\E_{c \sim f_{\mc{M}}(1,0)}[v_2(c)]$. 

Now, suppose John's utility function was instead
$$
u_2(c) = 
\begin{cases} 
  1, & c=\msf{n} \\
  -1, & c=\msf{g} \\
\end{cases}
$$
Since $v_2(\cdot) = u_2(\cdot) + 1$, both $v_2$ and $u_2$ are a positive affine transformation of one another, so both $u_2$ and $v_2$ represent the same preference $\succeq_j$. Under the same two states of the world $(1,0)$ and another where $(1,1)$,
\begin{align*}
    \E_{c \sim f_{\mc{M}}(1,0)}[u_2(c)] &= u_2(\msf{n})p_0 + u_2(\msf{g})(p_1 + p_2) \\
    & = (\varphi) + (-1)(e^{2\epsilon}\varphi + e^{\epsilon}\varphi) \\
    & = \varphi (1 - e^{2\epsilon} - e^{\epsilon})
\end{align*}
and
\begin{align*}
   \E_{c \sim f_{\mc{M}}(1,1)}[u_2(z)] &= u_2(\msf{n})p_0 + u_2(\msf{g})(p_1 + p_2) \\
    & = (1)(e^{\epsilon}\varphi) + (-1)(e^{2\epsilon}\varphi + \varphi)\\
    & = \varphi (e^{\epsilon} - e^{2\epsilon} - 1)
\end{align*}
Next, note that $-e^{3\epsilon} < -1$. But then, 
\begin{align*}
    e^{\epsilon}\E_{c \sim f_{\mc{M}}(1,0)}[u_2(c)] & = e^{\epsilon}\varphi (1 - e^{2\epsilon} - e^{\epsilon})\\
    & = \varphi (e^{\epsilon} - e^{3\epsilon} - e^{2\epsilon}) \\
    & < \varphi (e^{\epsilon} - 1 - e^{2\epsilon}) \\
    & = \E_{c \sim f_{\mc{M}}(1,1)}[u_2(c)]
\end{align*}
Thus, it is not the case that $\E_{c \sim f_{\mc{M}}(1,1)}[u_2(c)] \leq e^{\epsilon}\E_{c \sim f_{\mc{M}}(1,0)}[u_2(c)]$ despite the input differing by a single value.
This may be surprising to John who might suspect that his expected utility will change by at most a factor of $e^\epsilon$ whether he says he's a political dissident or not.
\end{ex}

\subsection{Measurement Types Matter: Ratio-Scales are Fragile for Interval Measures}
\label{subsec:fragile}

Example \ref{ex:hcua} shows that ratios of expected utilities can provide inconsistent conclusions depending on the particular choice of utility function. We saw that our analysis of John's actions were inconsistent, \emph{under the same differentially private mechanism}. At an intuitive level, this arises since $u_j$ is but one interval measure of $\succeq$ \cite{schoemaker1982expected}, so there is no true notion of an absolute zero for utility representations over $\R$ (after all, there are an infinite number of positive affine transformations available to rescale $u_j$ while preserving $\succeq_j$). Therefore, Proposition \ref{prop:ratio}, while mathematically true whenever $\text{Range}(u) \subseteq \R_{\ge 0}$, is in general a fragile way to study changes in utility under different differentially private mechanisms.

However, differences of expected utilities will produce consistent conclusions. To see why, recall that positive affine transformations respect $\succeq_j$ under expected utilities. That is, if $u_j$ represents $\succeq_j$ and $v_j(c) = au_j(c) + b$ for any $a >0 $ and any $b \in \R$, then $v_j$ represents $\succeq_j$ as well. 

Let $\Lambda$ be a mechanism (not necessarily $\epsilon$-differentially private), $x = (x_j,x_{-j})$, and $x' = (x_j',x_{-j})$. Consider $\Lambda(x)$ and $\Lambda(x')$ . Then,

$$
\E_{c\sim \Lambda(x)}[v_j(c)] - \E_{ c\sim \Lambda(x')}[v_j(c)] = a   \bigg( \E_{ c\sim \Lambda(x)}[u_j(c)] - \E_{ c\sim \Lambda(x')}(u_j(c)) \bigg)
$$

Since $a > 0$,  $\E_{ c\sim \Lambda(x)}[v_j(c)] \ge \E_{ c\sim \Lambda(x')}[v_j(c)] \iff \E_{ c\sim \Lambda(x)}[u_j(c)] \ge \E_{ c\sim \Lambda(x')}[u_j(c)]$, so we do not alter the ordering of expected utilities under a positive affine change of scale. Additionally, the magnitude of differences are preserved up to a positive constant, as 
$$
\bigg| \E_{c\sim \Lambda(x)}[v_j(c)] - \E_{ c\sim \Lambda(x')}[v_j(c)] \bigg| = a   \bigg| \E_{ c\sim \Lambda(x)}[u_j(c)] - \E_{ c\sim \Lambda(x')}[u_j(c)] \bigg|
$$

Therefore the classic Euclidean metric can be used as our yardstick to measure differences in a data subject’s utility under differential privacy, without fear that a positive affine change of scale would change the decision-theoretic analysis. 

\subsection{Worst-Case Expected Utility Differences Under Euclidean Distance}
\label{sec:util_diff_euclid}

To circumvent the non-negativity restriction placed on $u_j$ in Proposition \ref{prop:ratio}, we present an alternative characterization based on the worst-case difference in expected utilities that any agent can incur in under an $(\epsilon, \delta)$-differentially private mechanism.

\begin{thm} (Euclidean-Scale Utility Bound)
\label{thm:nominal}
If $\mc{M}$ satisfies $(\epsilon, \delta)$-differential privacy, then for all agents $j \in [n]$, finite $\mc{C}$, utility functions $u_j:\mc{C} \to \R$, functions $f: \mc{O} \rightarrow \mc{C}$, environments $x_{-j} \in \mc{X}_{-j}$, and responses $x_j, x'_j \in \mc{X}_j$,
$$
\bigg|\E_{c \sim f_{\mc{M}}(x_j,x_{-j})} [u_j(c)] - \E_{c \sim f_{\mc{M}}(x'_j,x_{-j})} [u_j(c)] \bigg| \le (e^{\epsilon} - 1 + \delta|\mc{C}|)\bigg(\max_c(u_j(c)) - \min_c(u_j(c))\bigg)
$$
\end{thm}
\begin{proof}
Since $u_j$ is bounded and $\mc{C}$ is finite, it achieves a minimum and maximum. For notational simplicity, denote $\min_c(|u_j(c)|)$ as $m_j$ and $\max_c(|u_j(c)|)$ as $M_j$. If $m_j = M_j$ the result trivially holds, so we only consider the case when $m_j < M_j$.

Let $v_j(\cdot) = \alpha u_j(\cdot) + \beta$, where $\alpha = (M_j - m_j)^{-1}$ and $\beta = -m_j \alpha$. Then $v_j(c) \in [0,1]$ for all $c \in \mc{C}$.

For any environment $x_{-j} \in \mc{X}_{-j}$ and any pair of responses $x_j, x'_j \in \mc{X}_j$,

\begin{align*}
    \E_{c \sim f_{\mc{M}}(x_j,x_{-j})} [v_j(c)] &= \sum_{c\in \mc{C}} v_j(c) \P\big(f_M(x_j,x_{-j}) = c\big)\\
    &\le \sum_{c\in \mc{C}} v_j(c)\bigg(e^{\epsilon}\P\big(f_M(x_j',x_{-j}) = c\big) + \delta\bigg) \\
    &= e^{\epsilon}\sum_{c\in \mc{C}} v_j(c)\P\big(f_M(x_j',x_{-j}) = c\big) + \delta\sum_{c\in \mc{C}} v_j(c)\P\big(f_M(x_j',x_{-j}) = c\big) \\
    &\le e^{\epsilon}\sum_{c\in \mc{C}} v_j(c)\P\big(f_M(x_j',x_{-j}) = c\big) + \delta\sum_{c\in \mc{C}} \P\big(f_M(x_j',x_{-j}) = c\big) \\
    &= e^{\epsilon}\E_{c \sim f_{\mc{M}}(x'_j,x_{-j})} [v_j(c)]  + \delta|C| \\
    &= (e^{\epsilon}-1)\E_{c \sim f_{\mc{M}}(x'_j,x_{-j})} [v_j(c)] + \E_{c \sim f_{\mc{M}}(x'_j,x_{-j})} [v_j(c)]  + \delta|C| \\
    & \le e^{\epsilon}-1 + \E_{c \sim f_{\mc{M}}(x'_j,x_{-j})} [v_j(c)]  + \delta|C|
\end{align*}
where the first inequality follows from the definition of differential privacy, and the remaining inequalities follow since $v_j(c) \le 1$ for all $c \in \mc{C}$. Rearranging terms, it follows that 
$$\E_{c \sim f_{\mc{M}}(x_j,x_{-j})} [v_j(c)]-\E_{c \sim f_{\mc{M}}(x'_j,x_{-j})} [v_j(c)] \le e^{\epsilon}-1+ \delta|C|$$ 
Using the definition of $v_j$ and the linearity of expectation, we deduce
$$
\E_{c \sim f_{\mc{M}}(x_j,x_{-j})} [u_j(c)]-\E_{c \sim f_{\mc{M}}(x'_j,x_{-j})} [u_j(c)] \le \alpha^{-1}\big(e^{\epsilon}-1+ \delta|C|\big)
$$
Alternatively, 
\begin{align*}
    \E_{c \sim f_{\mc{M}}(x_j,x_{-j})} [v_j(c)] &= \sum_{c\in \mc{C}} v_j(c) \P\big(f_M(x_j,x_{-j}) = c\big)\\
    &\ge \sum_{c\in \mc{C}} v_j(c)e^{-\epsilon}\big(\P\big(f_M(x_j',x_{-j}) = c\big) - \delta\big) \\
    &= e^{-\epsilon} \sum_{c\in \mc{C}} v_j(c)\P\big(f_M(x_j',x_{-j}) = c\big) - \delta\sum_{c\in \mc{C}} v_j(c)P\big(f_M(x_j',x_{-j}) = c\big)\\
    & \ge e^{-\epsilon} \sum_{c\in \mc{C}} v_j(c)\P\big(f_M(x_j',x_{-j}) = c\big) - \delta\sum_{c\in \mc{C}} P\big(f_M(x_j',x_{-j}) = c\big) \\
    &= e^{-\epsilon}\E_{c \sim f_{\mc{M}}(x'_j,x_{-j})} [v_j(c)]  - \delta|C| \\
    &= (e^{-\epsilon}-1)\E_{c \sim f_{\mc{M}}(x'_j,x_{-j})} [v_j(c)] + \E_{c \sim f_{\mc{M}}(x'_j,x_{-j})}[v_j(c)]  - \delta|C|\\
    &\ge (e^{-\epsilon}-1) + \E_{c \sim f_{\mc{M}}(x'_j,x_{-j})}  - \delta|C|\\
\end{align*}
where the first inequality follows by the definition of differential privacy, the second inequality follows since $v_j(c) \le 1$ for all $c \in \mc{C}$, and the last inequality follows as $v_j(c) \ge 0$ for all $c \in \mc{C}$. Hence, 
$$\E_{c \sim f_{\mc{M}}(x_j,x_{-j})} [v_j(c)]-\E_{c \sim f_{\mc{M}}(x'_j,x_{-j})} [v_j(c)] \ge e^{-\epsilon}-1 -\delta|C|$$ 
Using the definition of $v_j$ and the linearity of expectation yet again, we have
$$
\E_{c \sim f_{\mc{M}}(x_j,x_{-j})} [u_j(c)]-\E_{c \sim f_{\mc{M}}(x'_j,x_{-j})} [u_j(c)] \ge \alpha^{-1}\big(e^{-\epsilon}-1 -\delta|C|\big) = -\alpha^{-1}\big(-e^{-\epsilon}+1 +\delta|C|\big)
$$
Since $\epsilon \ge 0$, we have $e^{\epsilon} + e^{-\epsilon} \ge 2$ by the first derivative test. Indeed, let $h(z)= e^{z} + e^{-z} - 2$. Then $h(0) =0$ and $h'(z) = e^{z} - e^{-z} > 0$ for all $z > 0$, implying $h(z) \ge 0$ for all $z \ge 0$. Replacing $z$ with $\epsilon$ yields  $e^{\epsilon} + e^{-\epsilon} \ge 2$. Algebraic manipulation produces $e^{\epsilon} - 1 \ge -e^{-\epsilon} + 1$, which implies
$$
\E_{c \sim f_{\mc{M}}(x_j,x_{-j})} [u_j(c)]-\E_{c \sim f_{\mc{M}}(x'_j,x_{-j})} [u_j(c)] \ge -\alpha^{-1}\big(e^{\epsilon}-1 +\delta|C|\big)
$$
Invoking the definition of absolute value and $\alpha$ yields
$$
\bigg|\E_{c \sim f_{\mc{M}}(x_j,x_{-j})} [u_j(c)] - \E_{c \sim f_{\mc{M}}(x'_j,x_{-j})} [u_j(c)] \bigg| \le (e^{\epsilon} - 1 + \delta|\mc{C}|)\bigg(\max_c(u_j(c)) - \min_c(u_j(c))\bigg)
$$
\end{proof}

Taking a moment to interpret this, we see that regardless of the environment $x_{-j}$, the difference in expected utilities that any agent can incur is bounded, so long as $|\mc{C}|,\epsilon,\delta$, and $u_j$ are finite. As a sanity check, when both $\epsilon$ and $\delta$ are 0, the right-hand side of the inequality is 0, implies the expected utilities are the same regardless of $x = (x_1,...,x_n)$. This matches intuition, as such a mechanism does not use any information from any $x$, so the expected utility does not change. Moving away from the case where both $\epsilon$ and $\delta$ are both 0, we see that this upper bound is increasing in both of these quantities. This matches intuition as well, as allowing the probability of mechanism outputs to deviate more based on input $x$'s could alter agent $j$'s expected utilities.

\section{Incentivizing Responses for Standard Utility Functions}
\label{sec:finite}

We now turn our attention to incentivizing agents to join a private study. We consider the following setting. A surveyor wants to conduct a survey amongst $n$ people. Each agent $j$ can choose to opt-out by reporting $\perp$, or they can submit a value from $a \in \mc{X}_j - \{ \perp \}$. At the time of reporting, the agent's do not know the reports of the other agents (i.e., they do not have any knowledge of the environment), nor do they know the utility functions of other agents. 

The surveyor is interested in trying to recruit as many participants to the survey as possible. However, the surveyor may not force anyone to respond and instead must select a survey mechanism that induce voluntary participation. In the language of mechanism design, this condition is known as the voluntary participation constraint (also commonly referred to as the individual rationality constraint ) \cite{jackson2014mechanism}.

\begin{defn}
    A mechanism $\mc{M}$ satisfies the \textit{voluntary participation constraint} if for every agent $j \in [n]$, there exists some $a \in \mc{X}_j - \{ \perp \}$ such that for all $x_{-j} \in \mc{X}_{-j}$,
    $$
    \E_{c \sim f_{\mc{M}}(a, x_{-j})}[u_j(c)] \ge \E_{c \sim f_{\mc{M}}(\perp, x_{-j})}[u_j(c)]
    $$
\end{defn}

As we show in Example \ref{ex:ir}, the voluntary participation constraint is not always satisfied by every differentially private mechanism. Before presenting the example, we introduce the Laplace mechanism from differential privacy.

\begin{defn}
    The Laplace distribution with scale $\lambda > 0$, denoted as $Lap(\lambda)$, has probability density function $p(z)  = (2\lambda)^{-1}\exp(-|z| / \lambda)$ for all $z \in \R$.
\end{defn}

\begin{lem} (Laplace Mechanism, Theorem 3.6 of \cite{dwork2014algorithmic})
    Let $q: \mc{X} \to \R$ be a function. Define the sensitivity of $q$ as 
    $$GS(q) = \sup_{x,x' \text{neighbors}} |q(x) - q(x')|$$ If $GS(q) \in (0,\infty)$, then $M(x) = q(x) + L$ satisfies $\epsilon$-differential privacy for $L \sim Lap(GS(q)/\epsilon)$.
\end{lem}

\begin{ex}
\label{ex:ir}
Consider $\mc{X}_j = \{\perp,1,2\}$, $\mc{C} = \{0,1\}$, and
$$
q(x) = \sum_{i \in [n]} x_i \mb{I}(x_i \ne \perp)
$$
The sensitivity of $q$ is 2, so $\mc{M}(x) = q(x) + L$ for $L \sim Lap(2/\epsilon)$ is $\epsilon$-differentially private. Also, suppose $f(o) = \mb{I}(o \le n-1) $, $u_j(c) = \mb{I}(c=1)$, and environment $x_{-j} = (1,...,1)$. Pick any $a \in \mc{X}_j - \{\perp\}$. Then $q(\perp, x_{-j}) \le q(a, x_{-j})$. So,
\begin{align*}
   \E_{c \sim f_{\mc{M}}(a, x_{-j})}[v_j(c,r)] & = \P(f_{\mc{M}}(a, x_{-j}) = 1) \\
   & = \P(\mc{M}(a, x_{-j}) \le n-1) \\
   & = \P(q(a, x_{-j}) + L \le n-1) \\
   & = \P(L \le n-1 - q(a, x_{-j})) \\
   & = \int_{-\infty}^{n-1 - q(a, x_{-j})} \frac{\epsilon}{4}\exp(-\epsilon|z|/2)dz\\
   & < \int_{-\infty}^{n-1 - q(\perp, x_{-j})} \frac{\epsilon}{4}\exp(-\epsilon|z|/2)dz\\
   & = \P(L \le n-1 - q(\perp, x_{-j})) \\
   & = \P(q(\perp, x_{-j}) + L \le n-1) \\
   & = \P(\mc{M}(\perp, x_{-j}) \le n-1) \\
   & = \P(f_{\mc{M}}(\perp, x_{-j}) = 1) \\
   & = \E_{c \sim f_{\mc{M}}(\perp, x_{-j})}[v_j(c,r)]
\end{align*}
Hence, $\mc{M}$ does not satisfy the voluntary participation constraint. Therefore, under the Laplace mechanism, this prospective study participant is strictly better off not joining the analysis. 
\end{ex}

Example \ref{ex:ir} automatically implies the following existence claim.

\begin{prop}
    There exists an $\epsilon$-differentially private mechanism $\mc{M}$, a post-processor $f$, a collection agents $[n]$, response spaces $\mc{X}_i$ for $i \in [n]$, and a real-valued utility function $u_j$ for some agent $j \in [n]$ such that the voluntary participation constraint does not hold.
\end{prop}

One possible way to incentivize agents to join the study is to compensate them \cite{mishra2014theory}. This compensation can be thought of as money, but in general it is any reward that increases an agent's expected utility. The surveyor can induce participation if there exists some $r_j \in \R$ and some report $a \in \mc{X}_j - \{\perp\}$ such that, for all environments $x_{-j} \in \mc{X}_{-j}$,
$$
\E_{c \sim f_{\mc{M}}(a, x_{-j})}[u_j(c)] + r_j \ge \E_{c \sim f_{\mc{M}}(\perp, x_{-j})}[u_j(c)]
$$

Theoretically, for the sorts of utility functions we're currently considering, this is always possible. In the case where the expected utility of opting in with $a$ exceeds the expected utility of opting out, we don't have to compensate the agent at all. In this case, $r_j=0$. Alternatively, in cases where the expected utility of opting out exceeds that of opting in, we can always incentivize $j$ to join the study, provided they have a bounded utility function. 

In fact, when $\mc{M}$ satisfies $(\epsilon,\delta)$-differential privacy, we can derive an analytic expression for $r_j$ in the worst-case. Set $r_j = (e^{\epsilon} - 1 + \delta|\mc{C}| )(\max_c(u_j(c)) - \min_c(u_j(c)))$. Then for every $a \in \mc{X}_j - \{\perp\}$ and for all environments $x_{-j} \in \mc{X}_{-j}$,
Theorem \ref{thm:nominal} implies
\begin{align*}
    \E_{c \sim f_{\mc{M}}(\perp, x_{-j})}[u_j(c)] - \E_{c \sim f_{\mc{M}}(a, x_{-j})}[u_j(c)] & \le \bigg| \E_{c \sim f_{\mc{M}}(\perp, x_{-j})}[u_j(c)] - \E_{c \sim f_{\mc{M}}(a, x_{-j})}[u_j(c)] \bigg| \\
    & \le (e^{\epsilon} - 1 + \delta|\mc{C}| )\bigg(\max_c(u_j(c)) - \min_c(u_j(c))\bigg) \\
    & = r_j
\end{align*}
which implies 
$$
\E_{c \sim f_{\mc{M}}(a, x_{-j})}[u_j(c)] + r_j \ge \E_{c \sim f_{\mc{M}}(\perp, x_{-j})}[u_j(c)]
$$
Thus, in either case we can find some real-valued $r_j$ to incentivize joining. This analysis yields the following proposition.

\begin{prop}
\label{prop:compensation}
For any bounded utility function $u_j$, the surveyor needs to compensate agent $j$ at most $r_j = (e^{\epsilon} - 1 + \delta|\mc{C}|)\big(\max_c(u_j(c)) - \min_c(u_j(c))\big)$ to induce participate.
\end{prop}

As $\epsilon \to 0^+$ and $\delta \to 0^+$, a surveyor using the incentive scheme suggested by Proposition~\ref{prop:compensation} has to compensate the agent less, as the privacy properties decrease the difference in the bounds on expected utilities between opting-in and opting-out. In particular, when compensation takes the form of money, the surveyor can pay a finite amount to induce participation under pure $\epsilon$, pure $\delta$, and approximate differential privacy (assuming unbounded rewards from money). 

Proposition~\ref{prop:compensation} enables the comparison of variants of differential privacy under the rubric of participation costs under the suggested scheme. Consider $(\epsilon, \delta)$ and $(\epsilon',\delta')$-differential privacy, where $\epsilon, \epsilon', \delta, \delta' \ge 0$. Using the compensation scheme $r_j$ from Proposition \ref{prop:compensation}, the participation cost of $(\epsilon, \delta)$-differential privacy is cheaper than $(\epsilon', \delta')$-differential privacy if and only if 

$$
e^{\epsilon} - 1 + \delta|\mc{C}| < e^{\epsilon'} - 1 + \delta'|\mc{C}| \iff e^{\epsilon} - e^{\epsilon'}  < |\mc{C}|(\delta' - \delta)
$$

This has three implications for the comparison of pure $\epsilon$, pure $\delta$, and approximate differential privacy. When a surveyor uses the compensation scheme $r_j$ from Proposition \ref{prop:compensation}:

\begin{enumerate}
    \item The pure $\epsilon$ guarantee makes it cheaper to induce participation compared to $(\epsilon',\delta')$-differential privacy if and only if  $e^{\epsilon} - e^{\epsilon'}  < |\mc{C}|\delta'$.
    \item The pure $\delta$ guarantee makes it cheaper to induce participation than $(\epsilon',\delta')$-differential privacy if and only if $ 1- e^{\epsilon'}  < |\mc{C}|(\delta' - \delta)$.
    \item  And lastly, the pure $\epsilon$ guarantee makes it cheaper to induce participation compared to pure $\delta'$ guarantee if and only if $e^{\epsilon} - 1 \le \delta'|\mc{C}| \iff \epsilon < \ln(1+\delta'|\mc{C}|)$.
\end{enumerate}

Therefore, when comparing variants of differential privacy under the rubric of participation costs under the compensation scheme $r_j$, we find the specific parameters value used will determine which variant is more cost effective.
However, since Proposition~\ref{prop:compensation} is merely a bound, more exact analyses may yield cheaper compensation. 

\section{Utility Analysis for Non-Standard Utility Functions}
\label{sec:inf}

In this section, we consider ``high stakes'' outcomes, which we model with utility functions that can take hyperreal values. To set the stage, we first must take a brief mathematical interlude into the world of hyperreal numbers and hyperreal utility functions.

\subsection{Hyperreal Numbers and Hyperreal Utility Functions}
\label{sec:hyper_background}

The hyperreal numbers $^{\star}\R$ are an extension of the real numbers $\R$ that introduce two types of non-standard quantities: \textit{infinitesimals} (numbers that are smaller than any real number other than zero, in absolute value) and \textit{unlimiteds} (numbers that are larger than any real number, in absolute value).\footnote{Other authors use the phrases \textit{infinities} in place of unlimiteds (e.g., \cite{davis2009introduction})} To keep our discussion self-contained, we present a sketch of a construction of the hyperreals following the presentation in Krakoff \cite{krakoff2015hyperreals}. Additionally, we only discuss aspects of the hyperreals necessary for our differential privacy and utility analysis. As such, we will focus on a limited number of the arithmetic properties of $^{\star}\R$. For a more thorough discussion of the hyperreals, we point the reader to works of Davis \cite{davis2009introduction}, Krakoff \cite{krakoff2015hyperreals}, Keisler \cite{keisler2013elementary}, and Goldblatt \cite{goldblatt2012lectures}. 

This construction of the hyperreals utilizes equivalence classes over infinite sequences in $\R$. Consider $\R^{\N}$, the set of infinite sequences over $\R$. Define the relation $\equiv \subseteq \R^{\N} \times \R^{\N}$ by $x\equiv y \iff \{n \in \N \text{ : } x_n = y_n\}$ is contained in a mathematical object known as an ultrafilter. To avoid certain set-theoretic nuances, it suffices for our discussion to conceptually view this relation $\equiv$ as declaring two sequences $x,y \in \R^{\N}$ are equivalent if and only if they disagree on finitely many terms. Then $\equiv$ is reflexive, symmetric, and transitive, and is hence an equivalence relation. Therefore $\equiv$ partitions $\R^{\N}$. For $x \in \R^{\N}$, denote the equivalence class of $x$ as $[x] = \{y \in \R^{\N} \text{ : } x \equiv y\}$. The hyperreals $^{\star}\R$ are defined as the set of equivalence classes of $\R^{\N}$ under $\equiv$. Namely, $^{\star}\R = \{[x] \text{ : } x \in \R^{\N}\} $.

Using this construction of $^{\star}\R$, we can now extend orderings and additional operations from $\R$. First, the relation $<$ can be extended from $\R$ to $^{\star}\R$. Conceptually speaking, for $x,y \in$ $^{\star}\R$, we say $x < y \iff$ $\{n \in \N \text{ : } x_n \ge y_n\}$ is finite. And second, arithmetic operations in $^{\star}\R$ can defined component-wise over representative sequences in $\R^{\N}$. Formally, for $x,y \in$ $^{\star}\R$, $x+y = [x_n + y_n]$ and $xy = [x_n y_n]$. Then $^{\star}\R$ is an ordered field. So in particular, arithmetic in $^{\star}\R$ enables distributivity of multiplication over addition.

With this construction, we can find numbers $\omega \in$ $^{\star}\R$ such that $\omega> r$ for all $r \in \R$. Such $\omega$ are called \textit{positive unlimiteds} \cite{keisler2013elementary}. For example, consider $\omega = (\omega_1,\omega_2,...)$ where $\omega_n = n^{th}$ prime number. Since there are infinitely many primes, for any real number $r$, the components of $(r,r,...)$ exceed or equal the components of $\omega$ finitely many times. Similarly, we can define the \textit{negative unlimiteds} as $\omega \in$ $^{\star}\R$ such that $\omega < r$ for all $r \in \R$. 

There are two consequences of the arithmetic defined above. First, if $p \in \R_{>0}$ and $\omega \in$ $^{\star}\R$ is a positive unlimited, then $p\omega$ is a positive unlimited. And second, if $\omega \in$ $^{\star}\R$ is a positive unlimited, then $-\omega$ is well-defined (and more precisely a negative unlimited).

Transitioning from hyperreal numbers back to decision theory, Herzberg proved the existence of hyperreal utility functions $u: \mc{X} \to$ $^{\star}\R$ that represent preferences $\succeq$ satisfying certain decision-theoretic axioms \cite{herzberg2011hyperreal}. Because of this existence result, we can now consider hyperreal utility functions in the context of differential privacy.

\subsection{Differential Privacy Under Hyperreal Utility Functions}
\label{sec:hyper_dp_util}

We begin by constructing a scenario where no amount of compensation can be used to incentivize and individual to joining a study.

\begin{ex}
\label{ex:hyper}
    Let $u(c) = -\omega\mb{I}(c=1)$ for some positive unlimited $\omega \in$ $^{\star}\R$. Consider $\mc{X}_j = \{\perp,1,2\}$, $\mc{C} = \{0,1\}$, and
$$
q(x) = \sum_{i \in [n]} x_i \mb{I}(x_i \ne \perp)
$$
The sensitivity of $q$ is 2, so $\mc{M}(x) = q(x) + L$ for $L \sim Lap(2/\epsilon)$ is $\epsilon$-differentially private. Also, suppose $f(o) = \mb{I}(o \ge n) $, and environment $x_{-j} = (1,\ldots,1)$. 

There's no reward real number $r$ large enough to recruit participants.  To see this, suppose there were some $r \in \R$ to recruit $j$ to submit $a \in \mc{X}_j - \{\perp\}$ to the mechanism.
The reward $r$ must be such that
\begin{align*} 
r & \ge \E_{c \sim f_{\mc{M}}(\perp, x_{-j})}[u_j(c)] - \E_{c \sim f_{\mc{M}}(x_j, x_{-j})}[u_j(c)] \\
& = -\omega\P(f_{\mc{M}}(\perp, x_{-j}) = 1) + \omega\P(f_{\mc{M}}(x_j, x_{-j}) = 1) \\
& = \omega\bigg(\P(f_{\mc{M}}(x_j, x_{-j}) = 1) - \P(f_{\mc{M}}(\perp, x_{-j}) = 1)\bigg) \\
& = \omega\bigg(\P(\mc{M}(x_j, x_{-j}) \ge n) - \P(\mc{M}(\perp, x_{-j}) \ge n))\bigg) \\
& = \omega\bigg(\P(n-1+a + L \ge n) - \P(n-1+L \ge n))\bigg) \\
& = \omega\bigg(\P(L \ge 1-a) - \P(L \ge 1))\bigg) \\
& = \omega\int_{1-a}^{1}\frac{\epsilon}{4}\exp(-\epsilon |z| / 2) \diff z
\end{align*} 
where the second equality follows from the distributivity of multiplication over addition in hyperreal arithmetic. The resulting integral in the last equality is some $p \in \R_{>0}$. Since $\omega$ is a positive unlimited and $p \in \R_{>0}$, $\omega p$ equals some other positive unlimited $\omega'$. But then $\omega' \le r$, contradicting the positive unlimitedness of $\omega'$. Thus, there is no real-value reward large enough to recruit participants under these conditions.
\end{ex}

Example \ref{ex:hyper} implies the following existence claim.

\begin{prop}
    There exists an $\epsilon$-differentially private mechanism $\mc{M}$, a post-processor $f$, a collection agents $[n]$,  response spaces $\mc{X}_i$ for $i \in [n]$, and a hyperreal utility function $u_j$ for some agent $j \in [n]$ such that agent $j$ will never participate in the study in the presence of any real-valued payment scheme.
\end{prop}

On the flip side, we present a setting in which pure $\epsilon$-differential privacy is the only variant of differential privacy that induces voluntary participation.

\begin{ex}[The Court of Maimonides]
\label{ex:maim}
Suppose that John from Example~\ref{ex:hcua} were to be facing the Judge Maimonides instead of HCUA. Maimonides was a legal scholar who said ``It is better and more satisfactory to acquit a thousand guilty persons than to put a single innocent one to death''~\cite[p.\,178]{volokh97upenn-lr}.  Our fictionalized Maimonides has taken his demand for certainty even further. Maimonides insists on absolute proof before he would be willing to punish John in any way. However, if Maimonides is certain that John is a political dissident, he will sentence the John to death, which we model as having utility $-\omega$ where $\omega \in$ $^{\star}\R$ is a positive unlimited.

John views his life as being worth more than any amount of compensation.  Thus, there's no amount of compensation that will get John to choose a privacy guarantee that leaves the door open to conviction with any probability other than $0$ (or, an infinitesimal, if we allowed such).

Absolute $\epsilon$-differential privacy will never enable Maimonides to conclude that John is a political dissident with complete certainty. To see why, suppose to the contrary that Maimonides has an arbitrary $\epsilon$-differentially private mechanism $\mc{M}$ that, when combined with the post-processor $f$ from Example \ref{ex:hcua}, could classify John as a political dissident with absolute certainty when John reports $x_j \in \mc{X}_j$. Further suppose that the $\epsilon$-differentially private mechanism $f_\mc{M}$ Maimonides uses to determine John's political dissidence is \textit{minimally responsive} in the following sense: for all agents $j$ there exists $a, b \in \mc{X}_j$ and there exists $c \in \mc{C}$ such that for all environments $x_{-j}$, $\P(f_\mc{M}(a, x_{-j}) = c) \ne \P(f_\mc{M}(b, x_{-j}) = c)$. This is a reasonable assumption, as otherwise $\P(f_\mc{M}(x_j, x_{-j}) = c)$ would be the same for all $x_j$ and all $c$, meaning the output of the mechanism would be independent of John's message.  

Recall from Example \ref{ex:hcua} that $\mc{C} = \{\msf{g},\msf{n}\}$. Since Maimonides can infer John's political dissidence with probability 1, $\P(f_\mc{M}(x_j, x_{-j}) = \msf{g}) = 1$. So then $\P(f_\mc{M}(x_j, x_{-j}) = \msf{n}) = 0$. By the definition of $\epsilon$-differential privacy, for every $x'_j \in \mc{X}_j$, 

$$
\P(f_\mc{M}(x'_j, x_{-j}) = \msf{n}) \le e^{\epsilon}\P(f_\mc{M}(x_j, x_{-j}) = \msf{n}) = 0
$$

which implies $\P(f_\mc{M}(x'_j, x_{-j})= \msf{n}) = 0$, so $\P(f_\mc{M}(x'_j, x_{-j} = \msf{g}) =1$. Since $\mc{X}_j$ is binary, we have

$$
\P(f_\mc{M}(0, x_{-j}) = c) = \P(f_\mc{M}(1, x_{-j}) = c)
$$
for all $c \in \mc{C}$, contradicting the minimal responsiveness of $f_\mc{M}$.

By a symmetric argument, we similarly deduce that Maimonides could not infer $\msf{n}$ with probability 1 as well. Hence, for all $x_{-j}$ and for all $c \in \mc{C}$, $\P(f_\mc{M}(0, x_{-j}) = c), \P(f_\mc{M}(1, x_{-j}) = c) \in (0,1)$.

No matter how the other participants respond and no matter which output it selects, there will always be some probability (perhaps a tiny one) that $f_\mc{M}$ produced the output $\msf{g}$ despite every participant claiming not to be a political dissident. Thus, while the output might make Maimonides become nearly certain that John is a political dissident, he'll never reach certainty.

Under these conditions, John should reject approximate differential privacy in favor of absolute $\epsilon$-differential privacy for any $\delta > 0$ (that binds) and any size of reward. Furthermore, any (finite) value for $\epsilon$ will suffice for John to escape punishment. Thus, despite the exacting nature of John's privacy needs, we can find an $\epsilon$-differentially private mechanism such that the survey can nevertheless be very accurate with high probability.
\end{ex}

If we view opting-out as the same message as reporting $0$ to the mechanism, and allow for every agent to share John's utility function, the Example \ref{ex:maim} yields the following existence claim.

\begin{prop}
    There exists a collection agents $[n]$,  response spaces $\mc{X}_i$ for $i \in [n]$, and a collection of hyperreal utility function $u_j$ such that 
(1) for all minimally responsive binary $\epsilon$-differentially private mechanism $f_\mc{M}$, the voluntary participation constraint holds, but
(2) for all $\delta > 0$, there exists a minimally responsive binary $(\epsilon, \delta)$-differentially private mechanism $f'_\mc{M}$ such that the voluntary participation constraint does not hold. 
 \end{prop}

\section{Conclusion}
\label{sec:conclusion}

In this manuscript, we used utility theory to compare pure $\epsilon$-differential privacy, approximate differential privacy, and pure $\delta$-differential privacy. While utility theory is but one way to examine the tradeoffs incur by study participants, even in simple models where an individual's utility is only a function over the consequence space $\mc{C}$, our results are nuanced, demonstrating the richness and complexity of the problem space. 

We began our inquiry by examining Proposition \ref{prop:ratio}, a commonly referenced expected utility bound in differential privacy. Using this bound, we considered a candidate measurement method to compare changes in expected utility: a ratio-scale. However, this ratio-scale was unable to handle the interval nature of all real-valued utility representations of $\succeq_j$. We demonstrated this in Example \ref{ex:hcua} by constructing two distinct utility functions $u_j$ and $v_j$ that both represented $\succeq_j$, yet yielded different decision-theoretic conclusions. To avoid results that are heavily influenced by artifacts of a particular utility function, we found that the Euclidean distance produced consistent comparisons across all pairs of real-valued utility functions $u_j$ and $v_j$ that represent $\succeq_j$. 

In the case when participant utility functions are real-valued and bounded, we found in Theorem \ref{thm:nominal} that the difference in expected utilities is bounded. However, there are situations where an individual may forgo voluntarily joining a pure $\epsilon$-differentially private study (e.g., Example \ref{ex:ir}). We then showed in Proposition \ref{prop:compensation} that we can always find a compensation amount to incentivize participation. Since this compensation scheme is determined by $e^{\epsilon}-1+\delta|\mc{C}|$, we found that no one variant of differential privacy was universally cheaper than another. In general, the cost-effectiveness of inducing participation depends heavily on the specific values of $\epsilon$ and $\delta$ used. Therefore, in situations where privacy harms are finite and individual's can be compensated for their involvement in a study, multiple configurations of $\epsilon$ and $\delta$ can be used to incentivize participation. 

In cases where participant utility functions are hyperreal, we showed in Example \ref{ex:hyper} that there are situations where no amount of compensation can be used to induce participation. On the other end of the spectrum, we showed in Example \ref{ex:maim} that there are situations where individuals will voluntarily join a private study only the privacy guarantee is pure $\epsilon$-differential privacy.

Our work leaves open many future avenues of inquiry. For example, future work may consider preferences not only over consequences, but also preferences over the types of privacy protection afforded  \cite{kohli2018epsilon, cummings2016empirical}. Additionally, future work may also consider the preferences of the beneficiaries of a study, which may differ from the participants of the study \cite{kohli2021leveraging}. In such cases, there may be interesting and non-trivial dynamics between study participants' preferences and study beneficiaries' preferences. Taken together, the examples and analysis in this manuscript speak to the contextual nuances that arise when using utility theory and economics to understand the effects that privacy-enhancing technologies have on participants.

\bibliographystyle{ieeetr}
\bibliography{references}

\end{document}